\documentclass[11pt]{article} \usepackage[normal]{optional}
\usepackage{amsmath}
\usepackage{amssymb}
\usepackage{ifpdf}
\usepackage{cite}
\usepackage{pageno}

\usepackage[text={6.5in,9in},centering]{geometry}
\opt{normal}{\usepackage{commands-tam,numbering-tam}}
\opt{submission}{
    \usepackage{commands-tam}
    \numberwithin{equation}{section}
    \numberwithin{theorem}{section}
    \spnewtheorem{cor}[theorem]{Corollary}{\bfseries}{\itshape}
    \spnewtheorem{lem}[theorem]{Lemma}{\bfseries}{\itshape}
    \spnewtheorem{prop}[theorem]{Proposition}{\bfseries}{\itshape}
    \spnewtheorem{obs}[theorem]{Observation}{\bfseries}{\itshape}
}

\renewcommand{\vec}[1]{\mathbf{#1}}

\usepackage{url}
\usepackage{amsfonts}
\opt{normal}{\usepackage{amsthm}}
\opt{normal}{\usepackage{fullpage}}

\opt{submission,final}{}
\opt{normal}{}

\newcommand{\ignore}[1]{}

\usepackage[normalem]{ulem}

\ifpdf

  \usepackage[pdftex]{epsfig}
  \usepackage[pdfstartview={XYZ null null 1.00}, pdfpagemode=None,colorlinks=true,citecolor=blue,linkcolor=blue]{hyperref}

\else
    \usepackage[dvips]{epsfig}
\usepackage[dvips,colorlinks=true,citecolor=blue,linkcolor=blue]{hyperref}

\fi

\begin{document}

\title{Leaderless deterministic chemical reaction networks}

\opt{normal}{
    \author{
        David Doty\thanks{California Institute of Technology, Pasadena, CA, USA, {\tt ddoty@caltech.edu}. This author was supported by the Molecular Programming Project under NSF grant 0832824 and by NSF grants CCF-1219274 and CCF-1162589.}
        \and
        Monir Hajiaghayi\thanks{University of British Columbia, Vancouver, BC, Canada, {\tt monirh@cs.ubc.ca}. 
				}
    }
    \date{}
}

\opt{submission}{
}
\newcommand{\calN}{\mathcal{N}}
\newcommand{\calE}{\mathcal{E}}
\newcommand{\calY}{\mathcal{Y}}
\newcommand{\leqs}{\leq_\mathrm{s}}
\newcommand{\hf}{\hat{f}}
\newcommand{\vc}{\vec{c}}
\newcommand{\vo}{\vec{o}}
\newcommand{\vp}{\vec{p}}
\newcommand{\vi}{\vec{i}}
\newcommand{\vx}{\vec{x}}
\newcommand{\vy}{\vec{y}}
\newcommand{\vw}{\vec{w}}
\newcommand{\vz}{\vec{z}}
\newcommand{\vu}{\vec{u}}
\newcommand{\vv}{\vec{v}}
\newcommand{\oX}{\overline{X}}
\renewcommand{\vb}{\vec{b}}
\newcommand{\bfr}{\mathbf{r}}
\newcommand{\bfp}{\mathbf{p}}

\def\longrightharpoonup{\relbar\joinrel\rightharpoonup}
\def\longleftharpoondown{\leftharpoondown\joinrel\relbar}
\def\longrightleftharpoons{\mathop{\vcenter{\hbox{\ooalign{\raise1pt\hbox{$\longrightharpoonup\joinrel$}\crcr\lower1pt\hbox{$\longleftharpoondown\joinrel$}}}}}}
\def\rxn{\mathop{\rightarrow}\limits}  
\def\lrxn{\mathop{\longrightarrow}\limits}
\def\revrxn{\mathop{\rightleftharpoons}\limits}
\def\lrevrxn{\mathop{\longrightleftharpoons}\limits}

\maketitle

\begin{abstract}
  This paper answers an open question of Chen, Doty, and Soloveichik~\cite{CheDotSol12}, who showed that a function $f:\N^k\to\N^l$ is deterministically computable by a stochastic chemical reaction network (CRN) if and only if the graph of $f$ is a semilinear subset of $\N^{k+l}$.
  That construction crucially used ``leaders'': the ability to start in an initial configuration with constant but non-zero counts of species other than the $k$ species $X_1,\ldots,X_k$ representing the input to the function $f$.
  The authors asked whether deterministic CRNs without a leader retain the same power.

  We answer this question affirmatively, showing that every semilinear function is deterministically computable by a CRN whose initial configuration contains only the input species $X_1,\ldots,X_k$, and zero counts of every other species.
  We show that this CRN completes in expected time $O(n)$, where $n$ is the total number of input molecules.
  This time bound is slower than the $O(\log^5 n)$ achieved in~\cite{CheDotSol12}, but faster than the $O(n \log n)$ achieved by the \emph{direct} construction of~\cite{CheDotSol12} (Theorem 4.1 in the latest online version of~\cite{CheDotSol12}), since the fast construction of that paper (Theorem 4.4) relied heavily on the use of a fast, error-prone CRN that computes arbitrary computable functions, and which crucially uses a leader.
\end{abstract}





\section{Introduction}
\label{sec-intro}

In the last two decades, theoretical and experimental studies in molecular programming have shed light on the problem of integrating logical computation with biological systems.
One goal is to re-purpose the \emph{descriptive} language of chemistry and physics, which describes how the natural world works, as a \emph{prescriptive} language of programming, which prescribes how an artificially engineered system \emph{should} work.
When the programming goal is the manipulation of individual molecules in a well-mixed solution, the language of chemical reaction networks (CRNs) is an attractive choice.
A CRN is a finite set of reactions such as $X+Y\to X+Z$ among abstract molecular species, each describing a rule for transforming reactant molecules into product molecules.

CRNs may model the ``amount'' of a species as a real number, namely its concentration (average count per unit volume), or as a nonnegative integer (total count in solution, requiring the total volume of the solution to be specified as part of the system).
The latter integer counts model is called ``stochastic'' because reactions that discretely change the state of the system are assumed to happen probabilistically, with reactions whose reactants have high molecular counts more likely to happen first than reactions whose molecular counts are smaller.
The computational power of CRNs has been investigated with regard to simulating boolean circuits~\cite{magnasco1997chemical}, neural networks~\cite{hjelmfelt1991chemical}, digital signal processing~\cite{riedelDSP2012}, and simulating bounded-space Turing machines with an arbitrary small, non-zero probability of error with only a polynomial slowdown~\cite{angluin2006fast}.
CRNs are even efficiently Turing-universal, again with a small, nonzero probability of error over all time~\cite{soloveichik2009robust}.
Certain CRN termination and producibility problems are undecidable~\cite{zavattaro2008termination,CooSolWinBru09}, and others are $\PSPACE$-hard~\cite{ThaCon12}.
It is also difficult to design a CRN to ``delay'' the production of a certain species~\cite{CondonHMT12,condon2012reachability}.
Using a theoretical model of DNA strand displacement, it was shown that any CRN can be transformed into a set of DNA complexes that approximately emulate the CRN~\cite{cardelli2011strand}.
Therefore even hypothetical CRNs may one day be reliably implementable by real chemicals.

While these papers focus on the stochastic behaviour of chemical kinetics, our focus is on CRNs with deterministic guarantees on their behavior.
Some CRNs have the property that they deterministically progress to a correct state, no matter the order in which reactions occur.
For example, the CRN with the reaction $X \to 2Y$ is guaranteed eventually to reach a state in which the count of $Y$ is twice the initial count of $X$, i.e., computes the function $f(x)=2x$, representing the input by species $X$ and the output by species $Y$.
Similarly, the reactions $X_1 \to 2Y$ and $X_2 + Y \to \emptyset$, under arbitrary choice of sequence of the two reactions, compute the function $f(x_1,x_2) = \max\{0,2x_1-x_2\}$.

Angluin, Aspnes and Eisenstat~\cite{AngluinAE06} investigated the computational behaviour of deterministic CRNs under a different name known as \emph{population protocols}~\cite{angluin2006passivelymobile}. They showed that the input sets $S\subseteq \mathbb{N}^k$ decidable by deterministic CRNs (i.e. providing ``yes'' or ``no'' answers by the presence or absence of certain indicator species) are precisely the \emph{semilinear} subsets of $\mathbb{N}^k$.\footnote{Semilinear sets are defined formally in Section~\ref{sec-prelim}. Informally, they are finite unions of ``periodic'' sets, where the definition of ``periodic'' is extended in a natural way to multi-dimensional spaces such as $\N^k$.}
Chen, Doty, and Soloveichik~\cite{CheDotSol12} extended these results to function computation and showed that precisely the semilinear functions (functions $f$ whose graph $\setr{(\vx,\vy) \in \N^{k+l}}{f(\vx)=\vy}$ is a semilinear set) are deterministically computable by CRNs.
We say a function $f: \mathbb{N}^k \rightarrow \mathbb{N}^l$ is \emph{stably} (a.k.a. \emph{deterministically}) computable by a CRN $\calC$ if there are ``input'' species $X_1, \ldots ,X_k$ and ``output'' species $Y_1, \ldots ,Y_l$ such that, if $\calC$ starts with $x_1, \ldots ,x_k$ copies of $X_1, \ldots ,X_k$ respectively, then with probability one, it reaches a count-stable configuration in which the counts of $Y_1,  \ldots  , Y_l$ are expressed by the vector $f(x_1, . . . , x_k)$, and these counts never again change~\cite{CheDotSol12}.

The method proposed in~\cite{CheDotSol12} uses some auxiliary ``leader'' species present initially, in addition to the input species.
To illustrate their utility, suppose that we want to compute function $f(x) = x + 1$ with CRNs. Using the previous approach, we have an input species $X$ (with initial count $x$), an output species $Y$ and an auxiliary ``leader'' species $L$ (with initial count 1).
The following reactions compute $f(x)$:
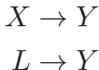
\begin{align*}
	X &\rightarrow Y\\
	L &\rightarrow Y
\end{align*}

However, it is experimentally difficult to prepare a solution with a single copy (or a small constant number) of a certain species.
The authors of~\cite{CheDotSol12} asked whether it is possible to do away with the initial ``leader'' molecules, i.e., to require that the initial configuration contains initial count $x_1,x_2,\ldots,x_k$ of input species $X_1,X_2,\ldots,X_k$, and initial count 0 of every other species.
It is easy to ``elect'' a single leader molecule from an arbitrary initial number of copies using a reaction such as $L+L \to L$, which eventually reduces the count of $L$ to 1.
However, the problem with this approach is that, since $L$ is a reactant in other reactions, there is no way in general to prevent $L$ from participating in these reactions until the reaction $L+L \to L$ has reduced it to a single copy.

Despite these difficulties, we answer the question affirmatively, showing that each semilinear function can be computed by a ``leaderless'' CRN, i.e., a CRN whose initial configuration contains only the input species.
To illustrate one idea used in our construction, consider the function $f(x)= x + 1$ described above.
In order to compute the function without a leader (i.e., the initial configuration has $x$ copies of $X$ and 0 copies of every other species), the following reactions suffice:
\begin{align}
	X &\to B +  2Y \label{rxn-intro-example-1}
\\
	B + B &\to B + K \label{rxn-intro-example-2}
\\
	Y + K &\to \emptyset \label{rxn-intro-example-3}
\end{align}
Reaction~\ref{rxn-intro-example-1} produces $x$ copies of $B$ and $2 x$ copies of $Y$.
Reaction~\ref{rxn-intro-example-2} consumes all copies of $B$ except one, so reaction~\ref{rxn-intro-example-2} executes precisely $x-1$ times, producing $x - 1$ copies of $K$.
Therefore reaction~\ref{rxn-intro-example-3} consumes $x-1$ copies of output species $Y$, eventually resulting in $2x - (x-1) = x + 1$ copies of $Y$.
Note that this approach uses a sort of leader election on the $B$ molecules.

In Section~\ref{sec-results}, we generalize this example, describing a leaderless CRN construction to compute any semilinear function.
We use a similar framework to the construction of~\cite{CheDotSol12}, decomposing the semilinear function into a finite union of affine partial functions (linear functions with an offset; defined formally in Section~\ref{sec-prelim}).
We show how to compute each affine function with leaderless CRNs, using a fundamentally different construction than the affine-function computing CRNs of~\cite{CheDotSol12}.
This result, Lemma~\ref{lem-affine-fast}, is the primary technical contribution of this paper.
Next, in order to decide which affine function should be applied to a given input, we employ the leaderless semilinear predicate computation of Angluin, Aspnes, and Eisenstat~\cite{angluin2006fast}; this latter part of the construction is actually identical to the construction of~\cite{CheDotSol12}, but we include it because our time analysis is different.

Let $n=\|\vx\| = \|\vx\|_1 = \sum_{i=1}^{k} \vx(i)$ be the number of molecules present initially, as well as the volume of the solution.
The authors of~\cite{CheDotSol12} showed, for each semilinear function $f$, a direct construction of a CRN that computes $f$ (using leaders) on input $\vx$ in expected time $O(n \log n)$.
They then combined this direct, error-free construction in parallel with a fast ($O(\log^5 n)$) but error-prone CRN that uses a leader to compute \emph{any} computable function (including semilinear), using the error-free computation to change the answer of the error-prone computation only if the latter is incorrect.
This combination speeds up the computation from expected time $O(n \log n)$ for the direct construction to expected time $O(\log^5 n)$ for the combined construction.

Since we assume no leaders may be supplied in the initial configuration, and since the problem of computing arbitrary computable functions without a leader remains a major open problem~\cite{angluin2006fast}, this trick does not work for speeding up our construction.
However, we show that with some care in the choice of reactions, the direct stable computation of a semilinear function can be done in expected time $O(n)$, improving upon the $O(n \log n)$ bound of the direct construction of~\cite{CheDotSol12}.



\section{Preliminaries}
\label{sec-prelim}


Given a vector $\vx\in\N^k$, let $\|\vx\| = \|\vx\|_1 = \sum_{i=1}^k |\vx(i)|$, where $\vx(i)$ denotes the $i$th coordinate of $\vx$.
A set $A \subseteq \N^k$ is \emph{linear} if there exist vectors $\vec{b},\vec{u}_1,\ldots,\vec{u}_p \in\N^k$ such that
$$A=\setl{\vec{b} + n_1 \vec{u}_1 + \ldots + n_p \vec{u}_p}{n_1,\ldots,n_p\in\N}.$$
$A$ is \emph{semilinear} if it is a finite union of linear sets.
If $f:\N^k\to\N^l$ is a function, define the \emph{graph} of $f$ to be the set
$\setl{ (\vx,\vy) \in \N^k \times \N^l }{ f(\vx) = \vy }.$
A function is \emph{semilinear} if its graph is semilinear.

We say a partial function $f:\N^k \dashrightarrow \N^l$ is \emph{affine} if there exist $kl$ rational numbers $a_{1,1},\ldots,a_{k,l}\in\Q$ and $l+k$ nonnegative integers $b_1,\ldots,b_l,c_1,\ldots,c_k\in\N$ such that, if $\vy = f(\vx)$, then for each $j \in \{1,\ldots,l\}$, $\vy(j) = b_j + \sum_{i=1}^k a_{i,j} (\vx(i) - c_i)$, and for each $i \in \{1,\ldots,k\}$, $\vx(i) - c_i \geq 0$.
In matrix notation, there exist a $k \times l$ rational matrix $\vec{A}$ and vectors $\vb \in \N^l$ and $\vec{c} \in \N^k$ such that $f(\vx) = \vec{A} (\vx-\vec{c}) + \vb$.

This definition of affine function may appear contrived; see~\cite{CheDotSol12} for an explanation of its various intricacies.
For reading this paper, the main utility of the definition is that it satisfies Lemma~\ref{lem-semilinear-function-finite-union-linear}.

Note that by appropriate integer arithmetic, a partial function $f:\N^k \dashrightarrow \N^l$ is affine if and only if there exist $kl$ integers $n_{1,1},\ldots,n_{k,l}\in\Z$ and $2l+k$ nonnegative integers $b_1,\ldots,b_l,c_1,\ldots,c_k,d_1,\ldots, \\ d_l\in\N$ such that, if $\vy = f(\vx)$, then for each $j \in \{1,\ldots,l\}$,
$\vy(j) = b_j + \frac{1}{d_j} \sum_{i=1}^k n_{i,j} (\vx(i) - c_i)$, and for each $i \in \{1,\ldots,k\}$, $\vx(i) - c_i \geq 0$.
Each $d_j$ may be taken to be the least common multiple of the denominators of the rational coefficients in the original definition.
We employ this latter definition, since it is more convenient for working with integer-valued molecular counts.

\subsection{Chemical reaction networks}

If $\Lambda$ is a finite set (in this paper, of chemical species), we write $\N^\Lambda$ to denote the set of functions $f:\Lambda \to \N$.
Equivalently, we view an element $\vc\in\N^\Lambda$ as a vector of $|\Lambda|$ nonnegative integers, with each coordinate ``labeled'' by an element of $\Lambda$.
Given $X\in \Lambda$ and $\vc \in \N^\Lambda$, we refer to $\vc(X)$ as the \emph{count of $X$ in $\vc$}.
We write $\vc \leq \vc'$ to denote that $\vc(X) \leq \vc'(X)$ for all $X \in \Lambda$.
Given $\vc,\vc' \in \N^\Lambda$, we define the vector component-wise operations of addition $\vc+\vc'$, subtraction $\vc-\vc'$, and scalar multiplication $n \vc$ for $n \in \N$.
If $\Delta \subset \Lambda$, we view a vector $\vc \in \N^\Delta$ equivalently as a vector $\vc \in \N^\Lambda$ by assuming $\vc(X)=0$ for all $X \in \Lambda \setminus \Delta.$

Given a finite set of chemical species $\Lambda$, a \emph{reaction} over $\Lambda$ is a triple $\alpha = \langle \bfr,\bfp,k \rangle \in \N^\Lambda \times \N^\Lambda \times \R^+$, specifying the stoichiometry of the reactants and products, respectively, and the \emph{rate constant} $k$.
If not specified, assume that $k=1$ (this is the case for all reactions in this paper), so that the reaction $\alpha=\langle \bfr,\bfp,1 \rangle$ is also represented by the pair $\pair{\bfr}{\bfp}.$
For instance, given $\Lambda=\{A,B,C\}$, the reaction $A+2B \to A+3C$ is the pair $\pair{(1,2,0)}{(1,0,3)}.$
A \emph{(finite) chemical reaction network (CRN)} is a pair $\calC=(\Lambda,R)$, where $\Lambda$ is a finite set of chemical \emph{species},
and $R$ is a finite set of reactions over $\Lambda$.
A \emph{configuration} of a CRN $\calC=(\Lambda,R)$ is a vector $\vc \in \N^\Lambda$.
We also write $\#_\vc X$ to denote $\vc(X)$, the \emph{count} of species $X$ in configuration $\vc$, or simply $\# X$ when $\vc$ is clear from context.

Given a configuration $\vc$ and reaction $\alpha=\pair{\bfr}{\bfp}$, we say that $\alpha$ is \emph{applicable} to $\vc$ if $\bfr \leq \vc$ (i.e., $\vc$ contains enough of each of the reactants for the reaction to occur).
If $\alpha$ is applicable to $\vc$, then write $\alpha(\vc)$ to denote the configuration $\vc + \bfp - \bfr$ (i.e., the configuration that results from applying reaction $\alpha$ to $\vc$).
If $\vc'=\alpha(\vc)$ for some reaction $\alpha \in R$, we write $\vc \to_\calC \vc'$, or merely $\vc \to \vc'$ when $\calC$ is clear from context.
An \emph{execution} (a.k.a., \emph{execution sequence}) $\calE$ is a finite or infinite sequence of one or more configurations $\calE = (\vc_0, \vc_1, \vc_2, \ldots)$ such that, for all $i \in \{1,\ldots,|\calE|-1\}$, $\vc_{i-1} \to \vc_{i}$.
If a finite execution sequence starts with $\vc$ and ends with $\vc'$, we write $\vc \to_\calC^* \vc'$, or merely $\vc \to^* \vc'$ when the CRN $\calC$ is clear from context.
In this case, we say that $\vc'$ is \emph{reachable} from $\vc$.


Turing machines, for example, have different semantic interpretations depending on the computational task under study (deciding a language, computing a function, etc.).
Similarly, in this paper we use CRNs to decide subsets of $\N^k$ (for which we reserve the term ``chemical reaction \emph{decider}'' or CRD) and to compute functions $f:\N^k\to\N^l$ (for which we reserve the term ``chemical reaction \emph{computer}'' or CRC).
In the next two subsections we define two semantic interpretations of CRNs that correspond to these two tasks.
We use the term CRN to refer to either a CRD or CRC when the statement is applicable to either type.

These definitions differ slightly from those of~\cite{CheDotSol12}, because ours are specialized to ``leaderless'' CRNs: those that can compute a predicate or function in which no species are present in the initial configuration other than the input species.
In the terminology of~\cite{CheDotSol12}, a CRN with species set $\Lambda$ and input species set $\Sigma$ is \emph{leaderless} if it has an \emph{initial context} $\sigma:\Lambda\setminus\Sigma \to \N$ such that $\sigma(S)=0$ for all $S\in\Lambda\setminus\Sigma.$
The definitions below are simplified by assuming this to be true of all CRNs.

We also use the convention of Angluin, Aspnes, and Eisenstat~\cite{AngluinAE06} that for a CRD, all species ``vote'' yes or no, rather than only a subset of species as in~\cite{CheDotSol12}, since this convention is convenient for proving time bounds.

\subsection{Stable decidability of predicates} \label{subsec-prelim-pred}
We now review the definition of stable decidability of predicates introduced by Angluin, Aspnes, and Eisenstat~\cite{AngluinAE06}.%
\footnote{Those authors use the term ``stably \emph{compute}'', but we reserve the term ``compute'' to apply to the computation of non-Boolean functions. Also, we omit discussion of the definition of stable computation used in the population protocols literature, which employs a notion of ``fair'' executions; the definitions are proven equivalent in~\cite{CheDotSol12}.}
Intuitively, the set of species is partitioned into two sets: those that ``vote'' yes and those that vote no,
and the system stabilizes to an output when a consensus vote is reached (all positive-count species have the same vote) that can no longer be changed (no species voting the other way can ever again be produced).
It would be too strong to characterize deterministic correctness by requiring all possible executions to achieve the correct answer; for example, a reversible reaction such as $A \longrightleftharpoons B$ could simply be chosen to run back and forth forever, starving any other reactions.
In the more refined definition that follows, the determinism of the system is captured in that it is impossible to stabilize to an incorrect answer, and the correct stable output is always reachable.

A \emph{(leaderless) chemical reaction decider} (CRD) is a tuple $\calD=(\Lambda,R,\Sigma,\Upsilon)$, where $(\Lambda,R)$ is a CRN, $\Sigma \subseteq \Lambda$ is the \emph{set of input species}, and $\Upsilon \subseteq \Lambda$ is the set of \emph{yes voters}, with species in $\Lambda \setminus \Upsilon$ referred to as \emph{no voters}.
An input to $\calD$ will be an \emph{initial configuration} $\vi\in\N^\Sigma$ (equivalently, $\vi\in\N^k$ if we write $\Sigma = \{X_1,\ldots,X_k\}$ and assign $X_i$ to represent the $i$'th coordinate); that is, only input species are allowed to be non-zero.
If we are discussing a CRN understood from context to have a certain initial configuration $\vi$, we write $\#_0 X$ to denote $\vi(X)$.

We define a global output partial function $\Phi:\N^\Lambda \dashrightarrow \{0,1\}$ as follows.
$\Phi(\vc)$ is undefined if either $\vc = \vec{0}$, or if there exist $S_0 \in \Lambda \setminus \Upsilon$ and $S_1\in\Upsilon$ such that $\vc(S_0) > 0$ and $\vc(S_1) > 0$.
Otherwise, either $(\forall S \in \Lambda) (\vc(S)>0 \implies S \in \Upsilon)$ or $(\forall S \in \Lambda) (\vc(S)>0 \implies S \in \Lambda\setminus\Upsilon)$; in the former case, the \emph{output} $\Phi(\vc)$ of configuration $\vc$ is 1, and in the latter case, $\Phi(\vc)=0$.

A configuration $\vo$ is \emph{output stable} if $\Phi(\vo)$ is defined and, for all $\vc$ such that $\vo \to^* \vc$, $\Phi(\vc) = \Phi(\vo)$.
We say a CRD $\calD$ \emph{stably decides} the predicate $\psi:\N^\Sigma \to \{0,1\}$ if,
for any initial configuration $\vi \in \N^k$, for all configurations $\vc\in\N^\Lambda$,
$\vi \to^* \vc$ implies $\vc\to^*\vo$ such that $\vo$ is output stable and $\Phi(\vo) = \psi(\vi)$.
Note that this condition implies that no incorrect output stable configuration is reachable from $\vi$.
We say that $\calD$ \emph{stably decides} a set $A\in\N^k$ if it stably decides its indicator function.

The following theorem is due to Angluin, Aspnes, and Eisenstat~\cite{AngluinAE06}:

\begin{theorem}[\cite{AngluinAE06}] \label{thm-semilinear}
A set $A \subseteq \N^k$ is stably decidable by a CRD if and only if it is semilinear.
\end{theorem}

The model they use is defined in a slightly different way; the differences (and those differences' lack of significance to the questions we explore) are explained in~\cite{CheDotSol12}.

\subsection{Stable computation of functions}

We now define a notion of stable computation of \emph{functions} similar to those above for predicates.
Intuitively, the inputs to the function are the initial counts of input species $X_1,\ldots,X_k$, and the outputs are the counts of output species $Y_1,\ldots,Y_l$.
The system stabilizes to an output when the counts of the output species can no longer change.
Again determinism is captured in that it is impossible to stabilize to an incorrect answer and the correct stable output is always reachable.

A \emph{(leaderless) chemical reaction computer (CRC)} is a tuple $\calC=(\Lambda,R,\Sigma,\Gamma)$, where $(\Lambda,R)$ is a CRN, $\Sigma \subset \Lambda$ is the \emph{set of input species}, $\Gamma \subset \Lambda$ is the \emph{set of output species}, such that $\Sigma \cap \Gamma = \emptyset$.
By convention, we let $\Sigma = \{X_1,X_2,\ldots,X_k\}$ and $\Gamma = \{Y_1,Y_2,\ldots,Y_l\}$.
We say that a configuration $\vo$ is \emph{output stable} if, for every $\vc$ such that $\vo \to^* \vc$ and every $Y_i \in \Gamma$, $\vo(Y_i) = \vc(Y_i)$ (i.e., the counts of species in $\Gamma$ will never change if $\vo$ is reached).
As with CRD's, we require initial configurations $\vi \in \N^\Sigma$ in which only input species are allowed to be positive.
We say that $\calC$ \emph{stably computes} a function $f:\N^k\to\N^l$ if for any initial configuration $\vi \in \N^{\Sigma}$,
$\vi \to^* \vc$ implies $\vc\to^*\vo$ such that $\vo$ is an output stable configuration with $f(\vi) = (\vo(Y_1), \vo(Y_2), \ldots, \vo(Y_l))$.
Note that this condition implies that no incorrect output stable configuration is reachable from $\vi$.

If a CRN stably decides a predicate or stably computes a function, we say the CRN is \emph{stable} (a.k.a. \emph{deterministic}).

%


\subsection{Kinetic model}

The following model of stochastic chemical kinetics is widely used in quantitative biology and other fields dealing with chemical reactions between species present in small counts~\cite{Gillespie77}.
It ascribes probabilities to execution sequences, and also defines the time of reactions, allowing us to study the computational complexity of the CRN computation in Section~\ref{sec-results}.

In this paper, the rate constants of all reactions are $1$, and we define the kinetic model with this assumption.
The rate constants do not affect the definition of stable computation; they only affect the time analysis.
Our time analyses remain asymptotically unaffected if the rate constants are changed (although the constants hidden in the big-$O$ notation would change).
A reaction is \emph{unimolecular} if it has one reactant and \emph{bimolecular} if it has two reactants.
We use no higher-order reactions in this paper.

The kinetics of a CRN is described by a continuous-time Markov process as follows.
Given a fixed volume $v \in \R^+$ and current configuration $\vc$, the \emph{propensity} of a unimolecular reaction $\alpha : X \to \ldots$ in configuration $\vc$ is $\rho(\vc, \alpha) = \#_\vc X$.
The propensity of a bimolecular reaction $\alpha : X + Y \to \ldots$, where $X \neq Y$, is $\rho(\vc, \alpha) = \frac{\#_\vc X \#_\vc Y}{v}$.
The propensity of a bimolecular reaction $\alpha : X + X \to \ldots$ is $\rho(\vc, \alpha) = \frac{1}{2} \frac{\#_\vc X (\#_\vc X - 1)}{v}$.
The propensity function determines the evolution of the system as follows.
The time until the next reaction occurs is an exponential random variable with rate $\rho(\vc) = \sum_{\alpha \in R} \rho(\vc,\alpha)$ (note that $\rho(\vc)=0$ if no reactions are applicable to $\vc$).


The kinetic model is based on the physical assumption of well-mixedness valid in a dilute solution.
Thus, we assume the \emph{finite density constraint}, which stipulates that a volume required to execute a CRN must be proportional to the maximum molecular count obtained during execution~\cite{SolCooWinBru08}.
In other words, the total concentration (molecular count per volume) is bounded.
This realistically constrains the speed of the computation achievable by CRNs.
Note, however, that it is problematic to define the kinetic model for CRNs in which the reachable configuration space is unbounded for some start configurations, because this means that arbitrarily large molecular counts are reachable.\footnote{One possibility is to have a ``dynamically" growing volume as in~\cite{SolCooWinBru08}.}
We apply the kinetic model only to CRNs with  configuration spaces that are bounded for each start configuration, choosing the volume to be equal to the reachable configuration with the highest molecular count (in this paper, this will always be within a constant multiplicative factor of the number of input molecules).

It is not difficult to show that if a CRN is stable and has a finite reachable configuration space from any initial configuration $\vi$, then under the kinetic model (in fact, for any choice of rate constants), with probability 1 the CRN will eventually reach an output stable configuration.

We require the following lemmas, which are proven in Appendix~\ref{sec-appendix}.

\begin{lem}\label{lem-unimol-log}
  \emph{Let $\mathcal{A} = \{A_1,\ldots,A_m\}$ be a set of species with the property that they appear only in applicable reactions of the form $A_i \to \sum_l B_l$, where $B_l \not\in \mathcal{A}$. Then starting from a configuration $\vc$ in which for all $i \in \{1,\ldots,m\}$, $\#_\vc A_i = O(n)$, with volume $O(n)$, the expected time to reach a configuration in which none of the described reactions can occur is $O(\log n)$.}
\end{lem}

\begin{lem}\label{lem-leader-election}
  \emph{Let $\mathcal{A} = \{A_1,\ldots,A_m\}$ be a set of species with the property that they appear only in applicable reactions of the form $A_i + A_j \to A_k + \sum_l B_l$, where $B_l \not\in \mathcal{A}$, and for all $i,j \in \{1,\ldots,m\}$, there is at least one reaction $A_i + A_j \to \ldots$. Then starting from a configuration $\vc$ in which for all $i \in \{1,\ldots,m\}$, $\#_\vc A_i = O(n)$, with volume $O(n)$, the expected time to reach a configuration in which none of the described reactions can occur is $O(n)$.}
\end{lem}

\begin{lem}\label{lem-populous-species-eats-other}
  \emph{Let $\mathcal{C} = \{C_1,\ldots,C_m\}$ and $\mathcal{A} = \{A_1,\ldots,A_p\}$ be two sets of species with the property that they appear only in applicable reactions of the form $C_i + A_j \to C_i + \sum_l B_l$, where $B_l \not\in \mathcal{A}$. Then starting from a configuration $\vc$ in which for all $i \in \{1,\ldots,m\}$, $\#_\vc A_i = O(n)$ and $\#_0 C_i = \Omega(n)$, with volume $O(n)$, the expected time to reach a configuration in which none of the described reactions can occur is $O(\log n)$.}
\end{lem}

\section{Leaderless CRCs can compute semilinear functions}
\label{sec-results}

To supply an input vector $\vx\in\N^k$ to a CRN, we use an initial configuration with $\vx(i)$ molecules of input species $X_i$.
Throughout this section, we let $n = ||\vx||_1 = \sum_{i=1}^k \vx(i)$ denote the initial number of molecules in solution.
Since all CRNs we employ have the property that they produce at most a constant multiplicative factor more molecules than are initially present, this implies that the volume required to satisfy the finite density constraint is $O(n)$.

Suppose the CRC $\calC$ stably computes a function $f:\N^k\dashrightarrow\N^l$.
We say that $\calC$ stably computes $f$ \emph{monotonically} if its output species are not consumed in any reaction.\footnote{Its output species could potentially be reactants so long as they are catalytic, meaning that the stoichiometry of the species as a product is at least as great as its stoichiometry as a reactant, e.g. $X + Y \to Z + Y$ or $A + Y \to Y + Y$.}

We show in Lemma~\ref{lem-affine-fast} that affine partial functions can be computed in expected time $O(n)$ by a leaderless CRC.
For its use in proving Theorem~\ref{thm-semilinear-function-n-no-leader}, we require that the output molecules be produced monotonically.
This is impossible for general affine partial functions.
For example, consider the function $f(x_1,x_2) = x_1 - x_2$ where $\dom f = \setl{(x_1,x_2)}{x_1 \geq x_2}$.
By withholding a single copy of $X_2$ and letting the CRC stabilize to the output value $\# Y = x_1-x_2+1$, then allowing the extra copy of $X_2$ to interact, the only way to stabilize to the correct output value $x_1-x_2$ is to consume a copy of the output species $Y$.
Therefore Lemma~\ref{lem-affine-fast} is stated in terms of an encoding of affine partial functions that allows monotonic production of outputs, encoding the output value $\vy(j)$ as the difference between the counts of two monotonically produced species $Y_j^P$ and $Y_j^C$, a concept formalized by the following definition.

Let $f:\N^k\dashrightarrow\N^l$ be a partial function.
We say that a partial function $\hf:\N^k\dashrightarrow\N^l\times\N^l$ is a \emph{diff-representation} of $f$ if $\dom f = \dom \hf$ and, for all $\vx\in\dom f$, if $(\vy_P,\vy_C) = \hf(\vx)$, where $\vy_P,\vy_C \in \N^l$, then $f(\vx) = \vy_P - \vy_C$, and $\vy_P = O(f(\vx))$.
In other words, $\hf$ represents $f$ as the difference of its two outputs $\vy_P$ and $\vy_C$, with the larger output $\vy_P$ possibly being larger than the original function's output, but at most a multiplicative constant larger.

The following lemma is the main technical result required for proving our main theorem, Theorem~\ref{thm-semilinear-function-n-no-leader}.
It shows that every affine function can be computed (via a diff-representation) in time $O(n)$ by a leaderless CRC.

\begin{lem} \label{lem-affine-fast}
  Let $f:\N^k\dashrightarrow\N^l$ be an affine partial function.
  Then there is a diff-representation $\hf:\N^k\dashrightarrow\N^l\times\N^l$ of $f$ and a leaderless CRC that monotonically stably computes $\hf$ in expected time $O(n)$.
\end{lem}

\begin{proof}
  If $f$ is affine, then there exist $kl$ integers $n_{1,1},\ldots,n_{k,l}\in\Z$ and $2l+k$ nonnegative integers $b_1,\ldots,b_l,c_1,\ldots,c_k,d_1,\ldots,d_l\in\N$ such that, if $\vy = f(\vx)$, then for each $j \in \{1,\ldots,l\}$,
  $\vy(j) = b_j + \frac{1}{d_j} \sum_{i=1}^k n_{i,j} (\vx(i) - c_i)$, and for each $i \in \{1,\ldots,k\}$, $\vx(i) - c_i \geq 0$.
  Define the CRC as follows.
  It has input species $\Sigma = \{X_1,\ldots,X_k\}$ and output species $\Gamma = \{Y^P_1,\ldots,Y^P_l,Y^C_1,\ldots,Y^C_l\}$.

  There are three main components of the CRN, separately handling the $c_i$ offset, the $n_{i,j}/d_j$ coefficient, and the $b_j$ offset.

  The latter two components both make use of $Y^C_j$ molecules to account for production of $Y^P_j$ molecules in excess of $\vy(j)$ to ensure that $\#_\infty Y^P_j - \#_\infty Y^C_j = \vy(j)$, which establishes that the CRC stably computes a diff-representation of $f$.
  It is clear by inspection of the reactions that $\#_\infty Y^P_j = O(\vy(j))$.

  Add the reaction
  \begin{equation} \label{rxn-X-to-C-and-Bs}
    X_1 \to C_{1,1} + B_1 + B_2 + \ldots + B_l + b_1 Y^P_1 + b_2 Y^P_2 + \ldots b_l Y^P_l
  \end{equation}
  The first product $C_{1,1}$ will be used to handle the $c_1$ offset, and the remaining products will be used to handle the $b_j$ offsets.
  For each $i \in \{2,\ldots,k\}$, add the reaction
  \begin{equation} \label{rxn-X-to-C}
    X_i \to C_{i,1}
  \end{equation}
  By Lemma~\ref{lem-unimol-log}, reactions~\eqref{rxn-X-to-C-and-Bs} and~\eqref{rxn-X-to-C} take time $O(\log n)$ to complete.

  We now describe the three components of the CRC separately.

  \begin{description}
  \item[\underline{$c_i$ offset}:]
      Reactions~\eqref{rxn-X-to-C-and-Bs} and~\eqref{rxn-X-to-C} produce $\vx(i)$ copies of $C_{i,1}$.
      We must reduce this number by $c_i$, producing $\vx(i)-c_i$ copies of $X'_i$, the species that will be used by the next component to handle the $n_{i,j}/d_j$ coefficient.
      A high-order reaction implementing this is $(c_i+1) C_{i,1} \to c_i C_{i,1} + X'_i$, since that reaction will eventually happen exactly $\vx(i)-c_i$ times (stopping when $\# C_{i,1}$ reaches $c_i$).
      This is implemented by the following bimolecular reactions.

      For each $i \in \{1,\ldots,k\}$ and $m,p \in \{1,\ldots,c_i\}$, if $m+p \leq c_i$, add the reaction
      $$
        C_{i,m} + C_{i,p} \to C_{i,m+p}.
      $$
      If $m+p > c_i$, add the reaction
      $$
        C_{i,m} + C_{i,p} \to C_{i,c_i} + (m+p-c_i) X'_i.
      $$
      By Lemma~\ref{lem-leader-election}, these reactions complete in expected time $O(n)$.

  \item[\underline{$n_{i,j} / d_j$ coefficient}:]
      For each $i \in \{1,\ldots,k\}$, add the reaction
      $$
        X'_i \to X_{i,1} + X_{i,2} + \ldots + X_{i,l}
      $$
      This allows each output to be associated with its own copy of the input.
      By Lemma~\ref{lem-unimol-log}, these reactions complete in expected time $O(\log n)$.

      For each $i \in \{1,\ldots,k\}$ and $j \in \{1,\ldots,l\}$, add the reaction
      $$
        X_{i,j} \to \left\
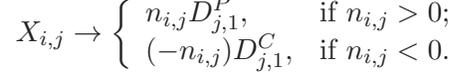
{
                      \begin{array}{ll}
                        n_{i,j} D_{j,1}^P, & \hbox{if $n_{i,j} > 0$;} \\
                        (- n_{i,j}) D_{j,1}^C, & \hbox{if $n_{i,j} < 0$.}
                      \end{array}
                    \right.
      $$
      By Lemma~\ref{lem-unimol-log}, these reactions complete in expected time $O(\log n)$.

      We must now divide $\# D_{j,1}^P$ and $\# D_{j,1}^C$ by $d_j$.
      This is accomplished by the high-order reactions $d_j D_{j,1}^P \to Y_j^P$ and $d_j D_{j,1}^C \to Y_j^C$.
      Similarly to the previous component, we implement these with the following reactions for $d_j \ge 1$.

      We first handle the case $d_j > 1$.
      For each $j \in \{1,\ldots,l\}$ and $m,p \in \{1,\ldots,d_j-1\}$, if $m+p \leq d_j-1$, add the reactions
      \begin{eqnarray*}
        D_{j,m}^P + D_{j,p}^P &\to& D_{j,m+p}^P
        \\
        D_{j,m}^C + D_{j,p}^C &\to& D_{j,m+p}^C
      \end{eqnarray*}
      If $m+p > c_i$, add the reactions
      \begin{eqnarray*}
        D_{j,m}^P + D_{j,p}^P &\to& D_{j,m+p-d_j}^P + Y^P_j
        \\
        D_{j,m}^C + D_{j,p}^C &\to& D_{j,m+p-d_j}^C + Y^C_j
      \end{eqnarray*}
	  By Lemma~\ref{lem-leader-election}, these reactions complete in expected time $O(n)$.

      When $d_j = 1$, we only have the following unimolecular reactions.
    	\begin{eqnarray*}
    		D_{j,1}^P &\to& Y^P_j\\
    		D_{j,1}^C &\to& Y^C_j
    	\end{eqnarray*}
      By Lemma~\ref{lem-unimol-log}, these reactions complete in expected time $O(\log n)$.

      These reactions will produce $\frac{1}{d_j} \sum_{n_{i,j} > 0} n_{i,j}(\vx(i)-c_i)$ copies of $Y^P_j$ and $- \frac{1}{d_j} \sum_{n_{i,j} < 0} n_{i,j}(\vx(i)-c_i)$ copies of $Y^C_j$.
      Therefore, letting $\#_\text{coef} Y^P_j$ and $\#_\text{coef} Y^C_j$ denote the number of copies of $Y^P_j$ and $Y^C_j$ eventually produced just by this component, it holds that $\#_\text{coef} Y^P_j - \#_\text{coef} Y^C_j = \frac{1}{d_j} \sum_{i=1}^k n_{i,j} (\vx(i) - c_i)$.

  \item[\underline{$b_j$ offset}:]
      For each $j \in \{1,\ldots,l\}$, add the reaction
      \begin{equation} \label{rxn-Bj}
        B_j + B_j \to B_j + b_j Y^C_j
      \end{equation}
      By Lemma~\ref{lem-leader-election}, these reactions complete in expected time $O(n)$.

      Reaction~\eqref{rxn-X-to-C-and-Bs} produces $b_j$ copies of $Y^P_j$ for each copy of $B_j$ produced, which is $\vx(i)$.
      Reaction~\eqref{rxn-Bj} occurs precisely $\vx(i)-1$ times.
      Therefore reaction~\eqref{rxn-Bj} produces precisely $b_j$ fewer copies of $Y^C_j$ than reaction~\eqref{rxn-X-to-C-and-Bs} produces of $Y^P_j$.
      This implies that when all copies of $Y^C_j$ are eventually produced by reaction~\eqref{rxn-Bj}, the number of $Y^P_j$'s produced by reaction~\eqref{rxn-X-to-C-and-Bs} minus the number of $Y^C_j$'s produced by reaction~\eqref{rxn-Bj} is $b_j$. \qedhere
  \end{description}
\end{proof}

We require the following lemma, proven in~\cite{CheDotSol12}.

\begin{lem}[\cite{CheDotSol12}]\label{lem-semilinear-function-finite-union-linear}
  Let $f:\N^k\to\N^l$ be a semilinear function.
  Then there is a finite set $\{f_1:\N^k \dashrightarrow \N^l,\ldots,f_m:\N^k \dashrightarrow \N^l\}$ of affine partial functions, where each $\dom f_i$ is a linear set, such that, for each $\vx\in\N^k$, if $f_i(\vx)$ is defined, then $f(\vx)=f_i(\vx)$, and $\bigcup_{i=1}^m \dom f_i = \N^k$.
\end{lem}

We require the following theorem, due to Angluin, Aspnes, and Eisenstat~\cite[Theorem 5]{angluin2006fast}, which states that any semilinear predicate can be decided by a CRD in expected time $O(n)$.

\begin{theorem}[\cite{angluin2006fast}]\label{thm-semilinear-predicate-n}
  Let $\phi:\N^k\to\{0,1\}$ be a semilinear predicate.
  Then there is a leaderless CRD $\calD$ that stably decides $\phi$, and the expected time to reach an output-stable configuration is $O(n)$.
\end{theorem}

The following is the main theorem of this paper.
It shows that semilinear functions can be computed by leaderless CRCs in linear expected time.

\newcommand{\hY}{\hat{Y}}

\begin{thm} \label{thm-semilinear-function-n-no-leader}
  Let $f:\N^k\rightarrow\N^l$ be a semilinear function.
  Then there is a leaderless CRC that stably computes $f$ in expected time $O(n)$.
\end{thm}

\begin{proof}
  The CRC will have input species $\Sigma = \{X_1, \ldots, X_k\}$ and output species $\Gamma = \{Y_1,\ldots,Y_l\}$.
  By Lemma~\ref{lem-semilinear-function-finite-union-linear}, there is a finite set $F=\{f_1:\N^k \dashrightarrow \N^l,\ldots,f_m:\N^k \dashrightarrow \N^l\}$ of affine partial functions, where each $\dom f_i$ is a linear set, such that, for each $\vx\in\N^k$, if $f_i(\vx)$ is defined, then $f(\vx)=f_i(\vx)$.
  We compute $f$ on input $\vx$ as follows.
  Since each $\dom f_i$ is a linear (and therefore semilinear) set, by Theorem~\ref{thm-semilinear-predicate-n} we compute each semilinear predicate $\phi_i = $ ``$\vx \in \dom f_i$ and $(\forall i' \in \{1,\ldots,i-1\})\ \vx \not\in \dom f_{i'}$?'' by separate parallel CRD's each stabilizing in expected time $O(n)$.
  (The latter condition ensures that for each $\vx$, precisely one of the predicates is true, in case the domains of the partial functions have nonempty intersection.)

  By Lemma~\ref{lem-affine-fast}, for each $i\in\{1,\ldots,m\}$, there is a diff-representation $\hat{f}_i$ of $f_i$ that can be stably computed by parallel CRC's.
  Assume that for each $i \in \{1,\ldots,m\}$ and each $j \in \{1,\ldots,l\}$, the $j$th pair of outputs $\vy_P(j)$ and $\vy_C(j)$ of the $i$th function is represented by species $\hY_{i,j}^P$ and $\hY_{i,j}^C$.
  We interpret each $\hY_{i,j}^P$ and $\hY_{i,j}^C$ as an ``inactive'' version of ``active'' output species $Y_{i,j}^P$ and $Y_{i,j}^C$.

  For each $i \in \{1,\ldots,m\}$, for the CRD $\calD_i=(\Lambda,R,\Sigma,\Upsilon)$ computing the predicate $\phi_i$, let $L_i^1$ represent any species in $\Upsilon$, and $L_i^0$ represent any species in $\Lambda\setminus\Upsilon$, and that once $\calD_i$ reaches an output stable configuration, $\# L_i^b = \Omega(n)$, where $b$ is the output of $\calD_i$.
  Then add the following reactions for each $i \in \{1,\ldots,m\}$ and each $j \in \{1,\ldots,l\}$:
  \begin{eqnarray}
    L_i^1 + \hY_{i,j}^P &\to& L_i^1 + Y_{i,j}^P + Y_j \label{rxn-activate-P-output}
    \\
    L_i^0 + Y_{i,j}^P &\to& L_i^0 + M_{i,j} \label{rxn-deactivate-P-output-1}
    \\
    M_{i,j} + Y_j &\to& \hY_{i,j}^P \label{rxn-deactivate-P-output-2}
  \end{eqnarray}
The latter two reactions implement the reverse direction of the first reaction -- using $L_i^0$ as a catalyst instead of $L_i^1$ -- using only bimolecular reactions.
Also add the reactions
  \begin{eqnarray}
    L_i^1 + \hY_{i,j}^C &\to& L_i^1 + Y_{i,j}^C \label{rxn-activate-C-output}
    \\
    L_i^0 + Y_{i,j}^C &\to& L_i^0 + \hY_{i,j}^C \label{rxn-deactivate-C-output}
  \end{eqnarray}
and
  \begin{eqnarray}
    Y_{i,j}^P + Y_{i,j}^C &\to& K_{j} \label{rxn-YP-YC-eat}
    \\
    K_{j} + Y_j &\to& \emptyset \label{rxn-K-Y-eat}
  \end{eqnarray}

  That is, a ``yes'' answer for function $i$ activates the $i$th output and a ``no'' answer deactivates the $i$th output.
  Eventually each CRD stabilizes so that precisely one $i$ has $L_i^1$ present, and for all $i' \neq i$, $L_{i'}^0$ is present.
  We now claim that at this point, all outputs for the correct function $\hat{f}_i$ will be activated and all other outputs will be deactivated.
  The reactions enforce that at any time, $\# Y_j = \# K_j + \sum_{i=1}^m (\# Y_{i,j}^P + \# M_{i,j})$.
  In particular, $\# Y_j \geq \# K_j$ and $\# Y_j \geq \# M_{i,j}$ at all times, so there will never be a $K_j$ or $M_{i,j}$ molecule that cannot participate in the reaction of which it is a reactant.
  Eventually $\# Y_{i,j}^P$ and $\# Y_{i,j}^C$ stabilize to 0 for all but one value of $i$ (by 
reactions~\eqref{rxn-deactivate-P-output-1}, \eqref{rxn-deactivate-P-output-2}, \eqref{rxn-deactivate-C-output}),
and for this value of $i$, $\# Y_{i,j}^P$ stabilizes to $\vy(j)$ and $\# Y_{i,j}^C$ stabilizes to 0 (by reaction \eqref{rxn-YP-YC-eat}).
  Eventually $\# K_j$ stabilizes to 0 by the last reaction.
  Eventually $\# M_{i,j}$ stabilizes to 0 since $L_i^0$ is absent for the correct function $\hat{f}_i$.
  This ensures that $\# Y_j$ stabilizes to $\vy(j)$.

  It remains to analyze the expected time to stabilization.
  Let $n = \|\vx\|$.
  By Lemma~\ref{lem-affine-fast}, the expected time for each affine function computation to complete is $O(n)$.
  Since the $\hY_{i,j}^P$ are produced monotonically, the most $Y_{i,j}^P$ molecules that are ever produced is $\#_\infty \hY_{i,j}^P$.
  Since
we have $m$ computations in parallel, the expected time for all of them to complete is $O(n m) = O(n)$ (since $m$ depends on $f$ but not $n$).
  We must also wait for each predicate computation to complete.
  By Theorem~\ref{thm-semilinear-predicate-n}, each of these predicates takes expected time $O(n)$ to complete, so all of them complete in expected time $O(m n) = O(n)$.

  At this point, the $L^i_1$ leaders must convert inactive output species to active, and $L^{i'}_0$ (for $i' \neq i$) must convert active output species to inactive.
  By Lemma~\ref{lem-populous-species-eats-other}, reactions~\eqref{rxn-activate-P-output},~\eqref{rxn-deactivate-P-output-1},~\eqref{rxn-activate-C-output}, and~\eqref{rxn-deactivate-C-output} complete in expected time $O(\log n)$.
  Once this is completed, by Lemma~\ref{lem-leader-election}, reaction~\eqref{rxn-deactivate-P-output-2} completes in expected time $O(n)$.
  Reaction~\eqref{rxn-YP-YC-eat} completes in expected time $O(n)$ by Lemma~\ref{lem-leader-election}.
  Once this is completed, reaction~\eqref{rxn-K-Y-eat} completes in expected time $O(n)$ by Lemma~\ref{lem-leader-election}.
\end{proof}


\section{Conclusion}
\label{sec-conclusion}

The clearest shortcoming of our leaderless CRC, compared to the leader-employing CRC of~\cite{CheDotSol12}, is the time complexity.
Our CRC takes expected time $O(n)$ to complete with $n$ input molecules, versus $O(\log^5 n)$ for the CRC of~\cite{CheDotSol12}.
The major open question is, for each semilinear function $f:\N^k\to\N^l$, is there a leaderless CRC that stably computes $f$ on input of size $n$ in expected time $t(n)$, where $t$ is a sublinear function?
This may relate to the question of whether there is a sublinear time CRN that solves the leader election problem, i.e., in volume $n$ with an initial state with $n$ copies of species $X$ and no other species initially present, produce a single copy of a species $L$.
However, it is conceivable that there is a direct way to compute semilinear functions quickly without needing to use a leader election.

If this is not possible for all semilinear functions, another interesting open question is to precisely characterize the class of functions that can be stably computed by a leaderless CRC in polylogarithmic time.
For example, the class of linear functions with positive integer coefficients (e.g., $f(x_1,x_2) = 3 x_1 + 2x_2$) has this property since they are computable by $O(\log n)$-time unimolecular reactions such as $X_1 \to 3Y, X_2 \to 2Y$.
However, most of the CRN programming techniques used to generalize beyond such functions seem to require some bimolecular reaction $A+B \to \ldots$ in which it is possible to have $\# A = \# B = 1$, making the expected time at least $n$ just for this reaction.

\paragraph{Acknowledgement.} We are indebted to Anne Condon for helpful discussions and suggestions.

{\small
\bibliographystyle{plain}
\bibliography{tam}
}

\newpage
\appendix
\section{Appendix}
\label{sec-appendix}

In this appendix, we prove some lemmas about the time complexity of certain common sequences of reactions.
Some of these are implicit or explicit in many earlier papers on stochastic CRNs, but we include them for the sake of self-containment.

The lemmas are stated with respect to a certain ``initial configuration'' $\vc$ that may not be the initial configuration of an actual CRN we define.
This is because the lemmas are employed to argue about CRNs that are guaranteed to evolve to some configuration $\vc$ that satisfies the hypothesis of the lemma, and we use the lemma to bound the time it takes for the CRN to complete a sequence of reactions, starting from $\vc$.
Therefore terms such as ``applicable reaction'' refer to being applicable from $\vc$ and any configuration reachable from it, although some additional inapplicable reactions may have been applicable prior to reaching the configuration $\vc$.

\begin{paragraph}{Lemma~\ref{lem-unimol-log}.}
  Let $\mathcal{A} = \{A_1,\ldots,A_m\}$ be a set of species with the property that they appear only in applicable reactions of the form $A_i \to \sum_l B_l$, where $B_l \not\in \mathcal{A}$. Then starting from a configuration $\vc$ in which for all $i \in \{1,\ldots,m\}$, $\#_\vc A_i = O(n)$, with volume $O(n)$, the expected time to reach a configuration in which none of the described reactions can occur is $O(\log n)$.
\end{paragraph}

\begin{proof}
  Assume the hypothesis.
  Let $c\in\N$ be the constant such that $\sum_{i=1}^m \#_\vc A_i \leq cn$.
  After each relevant reaction occurs, this sum is reduced by 1.
  Therefore no reactions can occur after $cn$ reactions have executed.
  If $\sum_{i=1}^m \# A_i = k$, the expected time for any reaction to occur is $\frac{1}{k}$.
  By linearity of expectation, the expected time for $cn$ reactions to execute is at most $\sum_{k=1}^{cn} \frac{1}{k} = O(\log n).$
\end{proof}

\begin{paragraph}{Lemma~\ref{lem-leader-election}.}
  Let $\mathcal{A} = \{A_1,\ldots,A_m\}$ be a set of species with the property that they appear only in applicable reactions of the form $A_i + A_j \to A_k + \sum_l B_l$, where $B_l \not\in \mathcal{A}$, and for all $i,j \in \{1,\ldots,m\}$, there is at least one reaction $A_i + A_j \to \ldots$. Then starting from a configuration $\vc$ in which for all $i \in \{1,\ldots,m\}$, $\#_\vc A_i = O(n)$, with volume $O(n)$, the expected time to reach a configuration in which none of the described reactions can occur is $O(n)$.
\end{paragraph}

\begin{proof}
  Assume the hypothesis.
  Let $c \in\N$ be a constant such that $\sum_{i=1}^m \#_\vc A_i \leq cn$, and let $c'$ be a constant such that the volume is at most $c' n$.
  After each relevant reaction occurs, this sum is reduced by 1.
  Therefore no reactions can occur after $cn-1$ reactions have executed.
	Now let $\rho(\vc,\alpha_{ij})$ be the propensity of the reaction $A_i + A_j \to A_k + \sum_l B_l$ which is equal to $\rho(\vc,\alpha_{ji})$ as well.
  Since $A_i$ can react with $A_j$ for any $i,j\in\{1,\ldots,m\}$, given that $\sum_{i=1}^m \# A_i = k$, the time for the next reaction to occur is an exponential random variable with rate equal to the sum of the rates of each possible reaction, i.e.,
 \begin{align*}
    \sum_{i=1}^m \sum_{\substack{j=1\\ j\ge i}}^m \rho(\vc,\alpha_{ij}) 
    &= 
    \frac{1}{2}\sum_{i=1}^m \sum_{\substack{j=1\\ j!=i}}^m \rho(\vc,\alpha_{ij}) + \sum_{i=1}^m \rho(\vc,\alpha_{ii})
    \\&=
    \frac{1}{2}\sum_{i=1}^m \sum_{\substack{j=1\\ j!=i}}^m \frac{\# A_i \#A_j}{c' n} +  \sum_{i=1}^m\frac{\# A_i (\#A_i-1)}{2 c' n}
    \\&= 
    \frac{1}{2c' n} \left[\sum_{i=1}^m \sum_{j=1}^m {\#A_i \#A_j} - \sum_{i=1}^m {\#A_i^2}\right] + \frac{1}{2c' n}\sum_{i=1}^m {\# A_i (\#A_i-1)}
    \\&=
    \frac{1}{2c' n} \left[ \sum_{i=1}^m \# A_i \left( \sum_{j=1}^m \#A_j \right) - \sum_{i=1}^m {\#A_i}\right]
    \\&=
    \frac{1}{2c' n}(k^2 - k)
\end{align*}
  so the expected time for the next reaction to occur is $\frac{c' n}{k^2-k}$.
  By linearity of expectation, the expected time for $cn-1$ reactions to execute is at most
  $
    \sum_{k=1}^{cn-1} \frac{c' n}{k^2-k}
    =
    c' n \sum_{k=1}^{cn-1}{(\frac{1}{k-1}-\frac{1}{k})}
    =
    c' n (1 - \frac{1}{cn-1})
	=
	O(n).
  $
\end{proof}

\begin{paragraph}{Lemma~\ref{lem-populous-species-eats-other}.}
  Let $\mathcal{C} = \{C_1,\ldots,C_m\}$ and $\mathcal{A} = \{A_1,\ldots,A_p\}$ be two sets of species with the property that they appear only in applicable reactions of the form $C_i + A_j \to C_i + \sum_l B_l$, where $B_l \not\in \mathcal{A}$. Then starting from a configuration $\vc$ in which for all $i \in \{1,\ldots,m\}$, $\#_\vc A_i = O(n)$ and $\#_0 C_i = \Omega(n)$, with volume $O(n)$, the expected time to reach a configuration in which none of the described reactions can occur is $O(\log n)$.
\end{paragraph}

\begin{proof}
  Assume the hypothesis.
  Then the counts of each $C_i$ do not decrease.
  (They may increase if some $B_l \in \mathcal{C}$, but this only strengths the conclusion.)
  Therefore this is similar to the proof of Lemma~\ref{lem-unimol-log}, since the expected time of each reaction when $\sum_{j=1}^p \#_0 A_j = k$ is within a constant of $\frac{1}{k}$.
\end{proof}

\end{document}